\documentclass[a4paper,10pt]{article}

\usepackage[utf8]{inputenc}
\usepackage{amsmath}
\usepackage{amsfonts}
\usepackage{amssymb}
\usepackage{amsthm}
\usepackage{xcolor,graphicx}
\usepackage{tabularx}
\usepackage{booktabs}
\usepackage{pdflscape}
\usepackage{xspace}
\usepackage{xspace}
\usepackage{cite}
\usepackage{dsfont}

\usepackage{rotating}
\usepackage[numbers,sort&compress]{natbib}

\usepackage[noend]{algorithmic}
\usepackage{algorithm}

\usepackage{authblk}

\usepackage{scrpage2}
\usepackage{a4wide}

\usepackage{tikz}
\usetikzlibrary{matrix,arrows,decorations.pathmorphing,decorations.pathreplacing,backgrounds,positioning,fit,shapes,arrows,calc,patterns}

\usepackage[pdfdisplaydoctitle,menucolor=orange!40!black,filecolor=magenta!40!black,urlcolor=blue!40!black,linkcolor=red!40!black,citecolor=green!40!black,colorlinks]{hyperref}

\usepackage{todonotes}
\usepackage[sort&compress,nameinlink,noabbrev,capitalize]{cleveref}

\theoremstyle{plain}
\newtheorem{theorem}{Theorem}
\newtheorem{lemma}[theorem]{Lemma}
\newtheorem{corollary}[theorem]{Corollary}

\newtheorem{observation}[theorem]{Observation}

\crefname{rrule}{Reduction Rule}{Reduction Rules}
\crefname{cond}{Condition}{Conditions}
\crefname{myclaim}{Claim}{Claims}
\crefname{step}{Step}{Steps}

\theoremstyle{definition}
\newtheorem{definition}{Definition}

\theoremstyle{remark}
\newtheorem{remark}{Remark}

\theoremstyle{plain}

\Crefname{ALC@unique}{Line}{Lines}
\crefname{cond}{Condition}{Conditions}
\crefname{step}{step}{steps}

\newcommand{\Match}{\textsc{Matching}\xspace}
\newcommand{\MaxCardMatch}{\textsc{Maximum Cardinality Matching}\xspace}

\newcommand{\distPara}{k}

\numberwithin{theorem}{section}

\newcommand{\TheTitle}{A Linear-Time Algorithm for Maximum-Cardinality \\ Matching on Cocomparability Graphs} 
\newcommand{\TheAuthors}{G.\ B.\ Mertzios, A.\ Nichterlein, R.\ Niedermeier}

\ifpdf
\hypersetup{
	pagebackref,
  pdftitle={\TheTitle},
  pdfauthor={\TheAuthors}
}
\fi

\begin{document}

\title{\TheTitle}

\author[1]{George B. Mertzios\thanks{Partially supported by the EPSRC grant EP/P020372/1.}}
\author[1,2]{André Nichterlein\thanks{Supported by a postdoc fellowship of the German Academic Exchange Service (DAAD) while at Durham University.}}
\author[2]{Rolf Niedermeier}

\affil[1]{Department of Computer Science, Durham University, UK \newline \texttt{george.mertzios@durham.ac.uk}}
\affil[2]{Algorithmics and Computational Complexity, Faculty~IV, TU Berlin, Germany \newline \texttt{\{andre.nichterlein,rolf.niedermeier\}@tu-berlin.de}}
\date{}

\maketitle

\thispagestyle{scrheadings}
\cfoot{}
\ohead{}
\ifoot{}

\begin{abstract}
Finding maximum-cardinality matchings in undirected graphs is arguably one of the most central graph problems. 
For general $m$-edge and $n$-vertex graphs, it is well-known to be solvable in $O(m\sqrt{n})$~time. 
We present a linear-time algorithm to find maximum-cardinality matchings on \emph{cocomparability} graphs, a prominent subclass of perfect graphs that strictly contains interval graphs as well as permutation graphs.
Our greedy algorithm is based on the recently discovered \emph{Lexicographic Depth First Search (LDFS)}. 
\end{abstract}

\section{Introduction}

The problem \Match (or \MaxCardMatch) is, given an undirected graph, to compute a maximum-cardinality set of disjoint edges. 
\Match is arguably among the most fundamental graph-algorithmic primitives that can be computed in polynomial time. 
More specifically, the asymptotically fastest known algorithm for computing a maximum-cardinality matching (subsequently called maximum matching) on an $n$-vertex and $m$-edge graph runs in~$O(m\sqrt{n})$ time~\cite{MV80}. 
No faster algorithm is known even when the given graph is bipartite~\cite{HK73}. 
Improving this running time, either on general graphs or on bipartite graphs, resisted decades of research. 
In terms of approximation, it is known that the~$O(m \sqrt{n})$ algorithm of Micali and Vazirani~\cite{MV80} implies a~$(1-\epsilon)$-approximation computable in~$O(m \epsilon^{-1})$ time~\cite{DP14}. 
For the weighted case, Duan and Pettie~\cite{DP14} provided a linear-time algorithm that computes a~$(1-\epsilon )$-approximate maximum-weight matching (the constant running time factor depending on~$\epsilon$ is $\epsilon^{-1}\log(\epsilon^{-1})$). 
In this work we take a route different to approximation and identify a large graph class, 
namely cocomparability graphs, on which we show that an optimal solution can be computed in linear time.

To identify more efficiently solvable special cases for finding maximum matchings has quite some history.
Yuster~\cite{Yus13} developed an algorithm with running time~$O(rn^2\log n)$,
where $r$ denotes the difference between maximum and minimum vertex degree 
of the input graph. Moreover, there are (quasi)\emph{linear-time} algorithms 
for computing maximum matchings in several special classes of 
graphs, including interval graphs~\cite{LR93}, 
convex bipartite graphs~\cite{SY96}, strongly chordal
graphs~\cite{DK98}, and chordal bipartite graphs~\cite{Cha96}. 
We refer to \cref{tbl:matching-algorithms} for a more thorough overview, also including results with superlinear running times.
See \cref{fig:graph-hierarchy} for an overview concerning the containment relation between the graph classes.

\begin{table}
	\caption{Fastest algorithms for \Match on special graph classes; $\omega < 2.373$ is the matrix multiplication exponent, that is, two~$n \times n$ matrices can be multiplied in~$O(n^\omega)$ time.}
	
	\begin{tabularx}{\textwidth}{ll}
		\toprule
		graph class & running time \\
		\midrule
		general 			& $O(m\sqrt{n})$~\cite{MV80}, $O(m \sqrt{n} \log(n^2/m) / \log (n))$~\cite{GK04}, $O(n^{\omega})$ (rand.) \cite{MS04} \\
		bipartite 			& $O(m \sqrt{n})$ \cite{HK73}, $O(n^{\omega})$ (rand.) \cite{MS04,San09}, $O(m^{1.43})$ \cite{Mad13} \\
		interval  			& $O(n \log n)$ (given an interval representation) \cite{LR93,MJ89} \\
		circular arcs		& $O(n \log n)$ \cite{LR93} \\
		co-interval 		& $O(n \log n + m)$ \cite{Gar03}\\
		convex 				& $O(n)$ \cite{SY96} \\
		planar	 			& $O(n^{\omega/2})$ (rand.) \cite{MS06} \\
		strongly chordal 	& $O(n+m)$ (given the strong perfect elimination order) \cite{DK98}  \\
		chordal bipartite 	& $O(n+m)$ \cite{Cha96} \\

		regular 	 		& $O(n^2 \log n)$ \cite{Yus13} \\
		cographs			& $O(n)$ (given a co-tree) \cite{YY93} \\
		\midrule
		co-comparability 	& $O(n + m)$ (\cref{thm:rmm-returns-max-matching} in \Cref{sec:cocomp-linear}) \\
		\bottomrule
	\end{tabularx}
	\label{tbl:matching-algorithms}
\end{table}

\begin{figure}[t!]
	\newcommand{\ColorOpen}{white}
	\newcommand{\ColorHard}{red!56}
	\newcommand{\ColorFPT}{green!50}
	\resizebox{\textwidth}{!}
	{\centering
		\begin{tikzpicture}
			\tikzstyle{para}=[rectangle, minimum height=.7cm,draw=black,rounded corners=3mm]

			\matrix[ampersand replacement=\&,row sep=0.5cm,column sep=0.5cm]
			{
				
				\node[para] (pla) {planar};
				\& 
				\& \node[para] (per) {perfect};
				\& 
				\& \\
				
				\node[para] (bip) {bipartite};
				\& \node[para] (ca) {circular arc};
				\& \node[para] (sco) {strongly chordal};
				\& \node[para] (coi) {co-interval};
				\& \node[para] (coc) {cocomparability};\\
				
				\node[para] (con) {convex};
				\& 
				\& \node[para] (int) {interval};
				\& 
				\& \node[para] (cog) {cograph};\\
			};
		\draw[thick] (per) to[out=190,in=20] (bip);
		\draw[thick] (per) -- (sco) -- (int);
		\draw[thick] (per) -- (coi);
		\draw[thick] (per) to[out=350,in=165] (coc);
		\draw[thick] (coc) to[out=200,in=20] (int);
		\draw[thick] (coc) -- (cog);
		\draw[thick] (ca) -- (int);
		\draw[thick] (bip) -- (con);
	\end{tikzpicture}
	}
	\caption{
		Overview over most of the graph classes mentioned in \cref{tbl:matching-algorithms}. An edge indicates that the class above strictly contains the class below.
	}
	\label{fig:graph-hierarchy}
\end{figure}
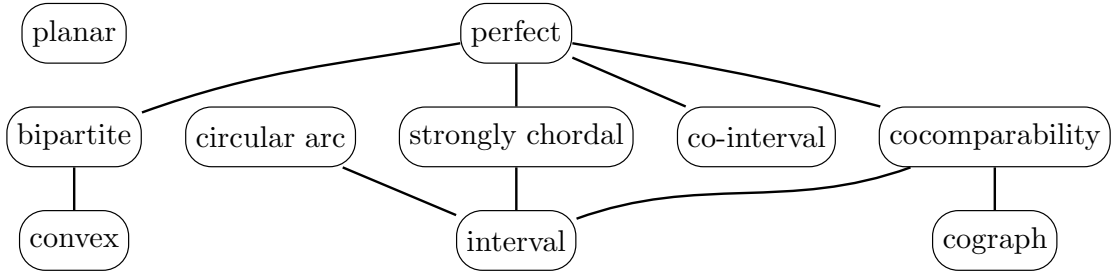

A graph $G$ is a \emph{cocomparability} graph if its complement~$\overline{G}$ admits a transitive orientation of its edges. 
These graphs (as well as their complements, i.e.~comparability graphs) arise naturally in several real-world applications as they are closely related to \emph{partially ordered sets} (also referred to as \emph{posets}). 
In particular, a given cocomparability graph $G$, together with a transitive orientation of the edges of its complement $\overline{G}$, can be equivalently represented by a poset.
Cocomparability graphs have been subject of intensive theoretical 
research~\cite{CL-ACTINF-97,CR-NETWORKS95,CDHKSIDMA16,CorneilDH-MPC,DS-SICOMP94,DHC16-ARXIV,DM-DAM17,KohlerM14,KratschStewart93,MertziosCorneil-LP,MertziosZ16}.
On the one hand, cocomparability graphs naturally generalize well-studied graph classes such as 
interval and permutation graphs~\cite{Golumbic04,BLS99}, 
\emph{trapezoid} (or \emph{bounded multitolerance}) graphs~\cite{MertziosC11,Ilic13,MaS94a,KrawczykW17,Takaoka15,Mertzios14}, 
\emph{parallelogram} (or \emph{bounded tolerance}) graphs~\cite{MertziosSZ09,MertziosSZ11,GiannopoulouM16}, 
\emph{triangle} (or \emph{PI$^*$}) graphs~\cite{Mertzios12,CorneilC87}, and 
\emph{simple-triangle} (or \emph{PI}) graphs~\cite{Mertzios15,Takaoka17,Takaoka-PI-recognition-ARXIV-18}. 
On the other hand, cocomparability graphs form an ``almost maximal'' subclass of perfect graphs~\cite{BLS99}.\footnote{For an overview of the relation between graph classes see~\url{http://www.graphclasses.org/}.}
Since perfect graphs (as well as comparability graphs) properly contain bipartite graphs (for which improving the~$O(m \sqrt{n})$ running time is a long-standing open question), it seems out of reach to obtain an algorithm for \Match with linear running time on perfect graphs. 
Consequently, designing a linear-time algorithm for cocomparability graphs provides a sharp boundary between $O(n+m)$-time algorithms and the known $O(m\sqrt{n})$-time algorithms for \Match.

\subparagraph*{Our contribution}
In this paper we present a linear-time algorithm for \Match on \emph{cocomparability} graphs. 
It is a simple greedy algorithm, referred to as \emph{Rightmost Matching (RMM)}, running on a specific vertex ordering.
Essentially the same greedy approach was earlier considered by Dragan~\cite{Dragan97} in the context of greedy matchable graphs.\footnote{Refer to Remark 2 in \Cref{subsec:cocomp-alg} for a discussion about the subtle but important differences to our approach.}
The vertex ordering is obtained by using (as a preprocessing step) the recently discovered \emph{Lexicographic Depth First Search (LDFS)} 
algorithm~\cite{Corneil-LDFS08}. 
Interestingly 
it turns out that RMM computes in a \emph{trivial} way a maximum matching on interval graphs, when applied on the standard interval graph vertex ordering\footnote{This is the vertex ordering that results from sorting the intervals according to their left endpoints. 
The RMM algorithm for interval graphs was discovered by Moitra and Johnson~\cite{MJ89}.}. 
Note that the class of interval graphs is a strict subset of the class of cocomparability graphs. 
So far a similar phenomenon of extending an interval graph algorithm to cocomparability graphs 
by using an LDFS preprocessing step has also been observed 
for the \textsc{Longest Path} problem~\cite{MertziosCorneil-LP}, 
the \textsc{Minimum Path Cover} problem~\cite{CorneilDH-MPC}, 
and the \textsc{Maximum Independent Set} problem~\cite{CDHKSIDMA16}.
Our results for the RMM algorithm, adding to the previous results~\cite{MertziosCorneil-LP,CorneilDH-MPC,CDHKSIDMA16}, 
provide evidence that cocomparability graphs present an ``interval graph structure'' when they are considered with an LDFS preprocessing step. 
This insight is of independent interest and might lead to new and more efficient combinatorial algorithms.

\subparagraph*{Preliminaries} %
We use standard notation from graph theory. 
In particular, all paths we consider are simple paths. 
A \emph{matching} in a graph is a set of pairwise disjoint edges.
Let~$G = (V,E)$ be an undirected graph and let~$M \subseteq E$ be a matching in~$G$.
A vertex~$v \in V$ is called \emph{matched} with respect to~$M$ if there is an edge in~$M$ containing~$v$, otherwise~$v$ is called \emph{free} with respect to~$M$.
If the matching~$M$ is clear from the context, then we omit ``with respect to~$M$''.
An \emph{alternating path} with respect to~$M$ is a path in~$G$ such that every second edge of the path belongs to~$M$.
An \emph{augmenting path} is an alternating path whose endpoints are free.
It is well-known that a matching~$M$ is maximum if and only if there is no augmenting path for it~\cite{LP86}. 
A graph $G=(V,E)$ is an \emph{interval} graph if we can assign to each vertex of $G$ a closed interval on the real line such that two vertices are adjacent in~$G$ if and only if the corresponding two intervals intersect. 
A \emph{comparability} graph is a graph whose edges can be transitively oriented, that is, if $u\rightarrow v$ (the edge $\{u,v\}$ is oriented towards~$v$) and $v\rightarrow w$, then $u\rightarrow w$. 
A \emph{cocomparability} graph $G$ is a graph whose complement~$\overline{G}$ is a comparability graph. 
The class of interval graphs is strictly included in the class of cocomparability graphs~\cite{BLS99}. 
Intuitively, we can transitively orient the ``non-edges'' of an interval graph, using the following ordering of non-intersecting intervals from left to right: 
Consider three intervals $I_a, I_b, I_c$ in an interval representation of an interval graph. 
If $I_a$ lies completely to the left of~$I_b$, and~$I_b$ lies completely to the left of~$I_c$, then also $I_a$ lies completely to the left of~$I_c$.

\section{A linear-time algorithm for cocomparability graphs} \label{sec:cocomp-linear}

To begin with, we present in \Cref{subsec:interval-alg} a simple greedy linear-time algorithm (called RMM) for computing a maximum matching~$M$ on interval graphs. 
Subsequently we provide in \Cref{subsec:vertex-orderings} all necessary background on vertex orderings for cocomparability graphs and on the Lexicographic Depth First Search (LDFS), which is needed for our algorithm on cocomparability graphs. 
Finally, as our central result, we prove in \Cref{subsec:cocomp-alg} that the algorithm RMM actually works also for cocomparability graphs.

\subsection{The greedy algorithm for interval graphs} \label{subsec:interval-alg}
Given an interval graph $G$ with $n$~vertices and $m$~edges, we first compute in $O(n+m)$~time an interval representation of~$G$ and, at the same time, we also sort the intervals according to their left endpoint~\cite{Olariu91}. 
The algorithm works as follows (cf.~\cite{LR93,MJ89}): 
\begin{enumerate}
	\item Initialize $M=\emptyset$ and label all vertices as ``unvisited''. 
	\item Pick the unvisited vertex (interval)~$x$ which has the rightmost left endpoint among all currently unvisited vertices in~$G$. 
	Then, label~$x$ as ``visited''. \label[step]{step:take-rightmost-unvisited}
	\item If~$x$ has at least one unvisited neighbor in~$G$, then pick the unvisited neighbor~$y$ of~$x$ which has the rightmost left endpoint among all unvisited neighbors of~$x$. 
	Then label~$y$ as ``visited'' and add the edge $\{x,y\}$ to~$M$. 
	\item If there is still an unvisited vertex in~$G$, then go to \Cref{step:take-rightmost-unvisited}. 
	\item Return~$M$.
\end{enumerate}
We call the above algorithm \emph{Rightmost-Matching (RMM)}. 
It can be executed in~$O(n+m)$ time; with a simple exchange argument we can show that the matching~$M$ returned by RMM is indeed maximum in $G$.
This algorithm implicitly uses the following vertex ordering that characterizes interval graphs.
It corresponds to sorting the intervals according to their left endpoints and can be computed in~$O(n+m)$ time from~$G$~\cite{Olariu91}.

\begin{lemma}[\hspace{-0.001cm}\cite{Olariu91}]
	\label[lemma]{lem:left-ordering-interval}
	$G=(V,E)$ is an interval graph if and only if there exists a vertex ordering $\sigma$ of $G$ (called an \emph{I-ordering}) such that, for all~$x <_{\sigma} y <_{\sigma} z$, if $\{x,z\} \in E$, then also~$\{x,y\} \in E$.
\end{lemma}

\subsection{Cocomparability graphs and vertex orderings} \label{subsec:vertex-orderings}

Before we proceed with our algorithm RMM and its analysis on cocomparability graphs (see~\Cref{subsec:cocomp-alg}), we now state vertex ordering characterizations of cocomparability graphs and of any vertex ordering that can result from an LDFS search on an arbitrary graph. 
The following vertex ordering characterizes cocomparability graphs~\cite{KratschStewart93}.

\begin{definition}[\hspace{-0.001cm}\cite{KratschStewart93}]
\label[definition]{def:umbrella-free}
	Let~$G=(V,E)$ be a graph. An ordering $\pi$ of the vertices~$V$ is an \emph{umbrella-free} ordering (or a \emph{CO-ordering}) if for all~$x <_{\pi} y <_{\pi} z$ it holds that if $\{x,z\} \in E$, then~$\{x,y\} \in E$ or~$\{y,z\} \in E$ (or both).
\end{definition}

\begin{lemma}[\hspace{-0.001cm}\cite{KratschStewart93}]
\label[lemma]{lem:cocomp-umbrella-free}
	A graph~$G = (V,E)$ is a cocomparability graph if and only if there exists an umbrella-free ordering~$\pi$ of~$V$.
\end{lemma}

Umbrella-free orderings directly generalize I-orderings for interval graphs (see~\Cref{lem:left-ordering-interval}).
It is worth noting here that, although there exists a linear-time algorithm to \emph{compute} 
an umbrella-free ordering $\pi$ of a given cocomparability graph~\cite{McConnellSpinrad94}, 
the fastest known algorithm to \emph{verify} that a given vertex ordering is indeed umbrella-free needs 
the same time as boolean matrix multiplication (Spinrad~\cite{Spinrad03} discusses this issue). 
As an example, we illustrate in \cref{ldfs-fig-antihole} the cocomparability graph 
$\overline{C_{6}}$, i.e.~the complement of the cycle on six vertices. 
In this graph, it is straightforward to check by~\Cref{def:umbrella-free} that 
the vertex ordering $\pi=(b,d,c,f,e,a)$ is indeed an umbrella-free ordering.

\begin{figure}[t]
\centering
\includegraphics[scale=0.9]{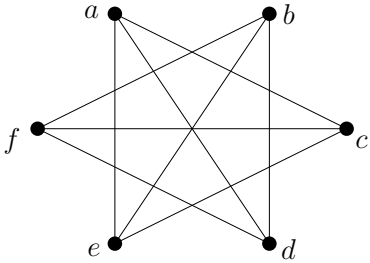}
\caption{The cocomparability graph $G=\overline{C_{6}}$, i.e.~the complement of a cycle 
on six~vertices. The vertex ordering $\pi=(b,d,c,f,e,a)$ is an umbrella-free ordering for $G$.}
\label{ldfs-fig-antihole}
\end{figure}

In the following we present the notion of a \emph{Lexicographic Depth First Search (LDFS) ordering} $\sigma$ (see~\Cref{def:ldfs-ordering}) due to Corneil and Krueger~\cite{Corneil-LDFS08}. 
This notion is based on \emph{good triples} and \emph{bad triples}, which are defined next.

\begin{definition}[\hspace{-0.001cm}\cite{Corneil-LDFS08}]\label[definition]{def:good-bad-triple}
	Let~$G = (V,E)$ be a graph and~$\sigma$ be an arbitrary ordering of~$V$. 
	Let~$a,b,c\in V$ be three vertices such that~$a <_{\sigma} b <_{\sigma} c$, $\{a,c\} \in E$, and ${\{a,b\} \notin E}$. 
	If there exists a vertex~$d$ such that~$a <_{\sigma} d <_{\sigma} b$, $\{d,b\} \in E$, and $\{d,c\} \notin E$, then $(a,b,c)$ is a \emph{good triple}, otherwise it is a \emph{bad triple}.
\end{definition}

\begin{definition}[\hspace{-0.001cm}\cite{Corneil-LDFS08}]\label[definition]{def:ldfs-ordering}
	Let~$G = (V,E)$ be a graph. 
	An ordering~$\sigma$ of~$V$ is an \emph{LDFS ordering} if~$\sigma$ has no bad triple.
\end{definition}

\begin{figure}[t]

\centering
\includegraphics[scale=0.9]{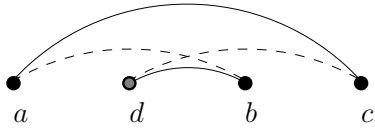}
\caption{A good triple $(a,b,c)$ and its vertex~$d$ as in \cref{def:good-bad-triple}, in the vertex ordering $\protect\sigma=(a,d,b,c)$. 
The edges $\{a,c\}$ and $\{d,b\}$ are indicated with solid lines and the non-edges $\{a,b\}$ and $\{d,c\}$ with dashed lines. 
Note that~$\{a,d\}$ and~$\{b,c\}$ can be edges or non-edges.
}
\label{good-triple-fig}
\end{figure}

An example of a good triple $(a,b,c)$ and the corresponding fourth vertex~$d$ is depicted in \cref{good-triple-fig}. 
Now we present the generic LDFS algorithm (\cref{ldfs-alg}) due to Corneil and Krueger~\cite{Corneil-LDFS08}. 
LDFS runs on an arbitrary connected graph $G$, starting at a distinguished vertex $u$. 
It is a variation of the well-known Depth First Search (DFS) algorithm; the main difference is that LDFS assigns labels to the vertices and uses the lexicographic order over these labels as a tie-breaking rule. 
Briefly, it proceeds as follows. 
Initially, the label $\varepsilon$ is assigned to every vertex. 
Then, iteratively, an unvisited vertex~$v$ with a lexicographically maximum label is chosen and removed from the graph. 
If $v$ is chosen as the $i$th vertex, then the label of each of its unvisited neighbors is being updated by \emph{prepending} the digit $i$ to it. 
Note that the digits in the label of any vertex are always in decreasing order. 
Hence all neighbors of the last chosen vertex have a lexicographically greater label than all its non-neighbors, and thus all vertices are visited in a depth-first search order.

\begin{algorithm}[htb]
\caption{LDFS($G,u$)~\cite{Corneil-LDFS08}} \label{ldfs-alg}
\begin{algorithmic}[1]
\REQUIRE{A connected graph $G=(V,E)$ with $n$ vertices and a vertex~$u\in V$.}
\ENSURE{An LDFS ordering $\sigma_u$ of the vertices of $G$.}

\medskip

\STATE{Assign the label $\varepsilon$ to all vertices and mark all vertices as unnumbered}
\STATE{label$(u) \leftarrow \{0\}$}

\FOR{$i=1$ to $n$}
     \STATE{Pick an unnumbered vertex $v$ with the lexicographically largest label} \label{ldfs-alg-star}
     \STATE{$\sigma_u(i) \leftarrow v$} \hfill\COMMENT{assign to $v$ the number $i$; $v$ is now numbered}
     \FOR{each unnumbered vertex $w\in N(v)$}
          \STATE{prepend $i$ to label$(w)$}
     \ENDFOR
\ENDFOR

\medskip

\RETURN{the ordering $\sigma_{u} = (\sigma_{u}(1),\sigma_{u}(2),\ldots,\sigma_{u}(n))$} 
\end{algorithmic}
\end{algorithm}

The execution of the LDFS algorithm is illustrated with the running example of \cref{ldfs-fig-antihole}. 
In this example, suppose that the LDFS algorithm starts at vertex $a$. 
Suppose that LDFS chooses vertex $c$ next. 
Now, ordinary DFS could choose either~$e$ or $f$ next, but LDFS has to choose~$e$, 
since $e$ has a greater label than $f$ ($e$ is a neighbor of the previously visited vertex $a$). 
The next visited vertex has to be $b$, since it is the only unvisited neighbor of $e$. 
The vertex following~$b$ in the LDFS ordering $\sigma_{a}$ must
be~$f$ rather than~$d$, since~$f$ has a greater label than~$d$ 
($f$~is a neighbor of vertex~$c$ which has been visited more recently than $d$'s neighbor $a$). 
Finally LDFS visits the last vertex $d$, completing the LDFS ordering as $\sigma_{a} = (a,c,e,b,f,d)$.

It is important here to connect the vertex ordering $\sigma_u$ that is returned by the LDFS algorithm (i.e.~\cref{ldfs-alg}) with the notion of an LDFS ordering, as defined in \cref{def:ldfs-ordering}. 
The next theorem due to Corneil and Krueger~\cite{Corneil-LDFS08} shows that a vertex ordering $\sigma$ of an arbitrary graph $G$ can be returned by an application of the LDFS algorithm to $G$ (starting at some vertex $u$ of $G$) if and only if $\sigma$ is an LDFS ordering.

\begin{theorem}[\hspace{-0.0005cm}\protect\cite{Corneil-LDFS08}]
For an arbitrary graph $G=(V,E)$, an ordering $\sigma$ of $V$ can be
returned by an application of \cref{ldfs-alg} to $G$ if and only if 
$\sigma$ is an LDFS ordering.
\end{theorem}

In the generic LDFS, there can be some choices to be made at \cref{ldfs-alg-star} of \cref{ldfs-alg}. 
More specifically, at some iteration there may be two or more vertices that have the same label; in this case the algorithm must break ties and choose one of these vertices. 
Generic LDFS (i.e.~\cref{ldfs-alg}) allows an arbitrary choice here. 
We present in the following a special type of an LDFS algorithm, called
LDFS$^{+}$ (see~\cref{ldfs+-alg} below), 
which chooses a specific vertex in such a case of equal labels, as follows. Along with
the graph $G=(V,E)$, an ordering $\pi$ of $V$ is also given as input. 
The algorithm LDFS$^{+}$ operates exactly as a generic LDFS that starts at the \emph{rightmost} vertex 
of $V$ in the ordering $\pi$, with the only difference that, in the case where 
at some iteration at least two unvisited vertices have the same label, LDFS$^{+}$ chooses the 
\emph{rightmost} vertex among them in the input ordering $\pi$. 
The resulting ordering is then denoted $\sigma ={}$LDFS$^{+}(G,\pi)$.

\begin{algorithm}[thb]
\caption{LDFS$^+$ ($G,\pi$)} \label{ldfs+-alg}
\begin{algorithmic}[1]
\REQUIRE{A connected graph $G=(V,E)$ with $n$ vertices and an ordering $\pi$ of $V$.}
\ENSURE{An LDFS ordering $\sigma$ of the vertices of $G$.}

\medskip

\STATE{Assign the label $\varepsilon$ to all vertices and mark all vertices as unnumbered}

\FOR{$i=1$ to $n$}
     \STATE{Pick the rightmost vertex $v$ in $\pi$ among the unnumbered vertices with the  lexicographically largest label} \label{ldfs+-alg-star}
     \STATE{$\sigma(i) \leftarrow v$} \hfill\COMMENT{assign to $v$ the number $i$; $v$ is now numbered}
     \FOR{each unnumbered vertex $w\in N(v)$}
          \STATE{prepend $i$ to label$(w)$}
     \ENDFOR
\ENDFOR

\medskip

\RETURN{the ordering $\sigma = (\sigma(1),\sigma(2),\ldots,\sigma(n))$}
\end{algorithmic}
\end{algorithm}

Consider our running example of \cref{ldfs-fig-antihole}. 
In this graph $G=\overline{C_6}$, suppose that LDFS$^+$ is given as input the 
umbrella-free ordering $\pi=(b,d,c,f,e,a)$. Then the ordering $\sigma ={}$LDFS$^{+}(G,\pi)$ 
is computed using \cref{ldfs+-alg} as follows. 
The first visited vertex is $a$, since $a$ is the rightmost vertex in the ordering $\pi$. 
Now, LDFS (see \cref{ldfs-alg}) could choose any of the neighbors $c,d,e$ of $a$ next, 
but LDFS$^+$ (see \cref{ldfs+-alg}) has to choose~$e$, 
since $e$ is the rightmost among these vertices in the ordering $\pi$. 
In this example there exists no further tie among vertices with the same label. 
Thus, proceeding similarly to our example vertex ordering for LDFS above, 
it follows that the resulting ordering $\sigma ={}$LDFS$^{+}(G,\pi)$ is $\sigma = (a,e,c,f,b,d)$.

For the purposes of our algorithm RMM for computing a maximum matching on cocomparability graphs in \Cref{subsec:cocomp-alg}, 
we will consider an arbitrary umbrella-free vertex ordering $\pi$ of the input cocomparability graph $G$, 
and we will then compute the LDFS ordering $\sigma ={}$LDFS$^{+}(G,\pi)$, 
by applying \cref{ldfs+-alg} (i.e.~LDFS$^{+}$) to $\pi$. 
Our RMM algorithm (see \cref{alg:RMM}) will then take this LDFS ordering $\sigma$ as input, 
together with the graph $G$.
It is important to note here that, starting from an umbrella-free ordering $\pi$, 
the LDFS vertex ordering $\sigma ={}$LDFS$^{+}(G,\pi)$ remains umbrella-free~\cite{CorneilDH-MPC}. 
That is, $\sigma$ satisfies both the conditions of \cref{def:umbrella-free} and 
\cref{def:ldfs-ordering}, and thus $\sigma$ is \emph{simultaneously} an LDFS ordering 
and an umbrella-free ordering. For this reason we refer to $\sigma$ as an \emph{LDFS umbrella-free} 
vertex ordering of the input cocomparability graph $G$.
Finally note that, given an umbrella-free ordering $\pi$ of a cocomparability graph $G$ with $n$ vertices 
and $m$ edges, the ordering LDFS$^{+}(G,\pi)$ can be computed in $O(n+m)$ time~\cite{KohlerM14}.

\subsection{The algorithm for cocomparability graphs} \label{subsec:cocomp-alg}

Once we have computed in $O(n+m)$ time the LDFS umbrella-free ordering~$\sigma ={}$LDFS$^{+}(G,\pi)$, 
we apply our simple linear-time algorithm \emph{Rightmost-Matching (RMM)}, see~\cref{alg:RMM}, 
to compute a new vertex ordering $\widehat{\sigma}$ and a maximum matching $M$ of $G$. 
RMM is a simple greedy algorithm which operates as follows. 
At every step it visits the rightmost unvisited vertex~$x$ in $\sigma$ and it labels $x$ as visited. 
Then, if $x$ does not have any unvisited neighbor, then RMM proceeds at the next step by visiting again 
the currently unvisited vertex in $\sigma$; 
note that this vertex is now different from~$x$, as~$x$ has been already labeled as visited. 
Otherwise, if~$x$ has at least one unvisited neighbor, then RMM visits after $x$ its rightmost unvisited 
neighbor~$y$ in the ordering $\sigma$ and it also adds the edge $\{x,y\}$ to the computed matching $M$. 

\begin{algorithm}[h!]
	\caption{RMM($G$, $\sigma$).}
	\label[algorithm]{alg:RMM}
	\begin{algorithmic}[1]
    \REQUIRE{A cocomparability graph $G$ with an LDFS umbrella-free ordering $\sigma$ of $G$.}
    \ENSURE{A vertex ordering $\widehat{\sigma}$ of $G$ and a maximum matching of $G$.}

    \medskip
	
		\STATE\label[line]{line:init}{Label all vertices ``unvisited''; \ $i \leftarrow 0$; \ $M \leftarrow \emptyset$ \;}
		\WHILE{there are unvisited vertices\label[line]{line:while}}
			\STATE{Pick the rightmost unvisited vertex~$x$ in~$\sigma$ and label~$x$ as ``visited'' \; \label[line]{line:pick-rightmost-x}}
			\STATE\label[line]{line:add-x}{$i \leftarrow i+1$; \ $\widehat{\sigma}(i) \leftarrow x$ \hfill \COMMENT{add vertex $x$ to the ordering $\widehat{\sigma}$}}
			\IF{$x$ has at least one unvisited neighbor\label[line]{line:if-unvisited-neighbor}} 
			\STATE\label[line]{line:pick-rightmost-y}{Pick the rightmost unvisited neighbor~$y$ of~$x$ and label~$y$ as ``visited'' \; }
			\STATE\label[line]{line:add-y}{$i \leftarrow i+1$; \ $\widehat{\sigma}(i) \leftarrow y$ \hfill \COMMENT{add vertex $y$ to the ordering $\widehat{\sigma}$}}
			\STATE\label[line]{line:match-xy}{$M \gets M \cup \{\{x,y\}\}$ \hfill \COMMENT{match $x$ and $y$}}
			\ENDIF
		\ENDWHILE
		\RETURN\label[line]{line:return} the ordering $\widehat{\sigma}$ and the matching $M$ 
	\end{algorithmic}
\end{algorithm}

\begin{remark}
Since any I-ordering of an interval graph is also an LDFS umbrella-free ordering (see~\cref{lem:left-ordering-interval,def:umbrella-free,def:good-bad-triple,def:ldfs-ordering}), 
note that~\cref{alg:RMM} also works
with an interval graph $G$ and an I-ordering~$\sigma$ of~$G$ as input.
In this case, RMM($G$,\,$\sigma$) is actually exactly the same RMM algorithm as we sketched in \Cref{subsec:interval-alg} for interval graphs. 
\end{remark}

\begin{remark}
Essentially the same greedy approach as our RMM algorithm was already considered by Dragan~\cite{Dragan97}.%
\footnote{With the only difference that Dragan's algorithm visits the vertices from left to right and always matches a vertex with its leftmost unvisited neighbor.}
More specifically, he characterized those graphs $G$ which admit a vertex ordering~$\tau$ such that the greedy algorithm computes a maximum matching on \emph{every induced subgraph}~$F$ of~$G$ when applied to the induced sub-ordering of $\tau$ on the vertices of $F$. 
These graphs $G$ having the above property are called \emph{greedy matchable} graphs~\cite{Dragan97}. 
We prove that cocomparability graphs admit a vertex ordering~$\sigma$ (namely an LDFS umbrella-free ordering) such that the greedy algorithm computes a maximum matching in the input graph $G$ itself (and \emph{not} in every induced subgraph of $G$). 
That is, Dragan~\cite{Dragan97} studied a problem that is very different 
from computing a maximum matching in a given graph.

Dragan~\cite{Dragan97} proved that greedy matchable graphs form a subclass of weakly triangulated graphs; a graph is weakly triangulated if it neither contains (as an induced subgraph) a chordless cycle of length at least five nor the complement of such a chordless cycle. 
On the contrary, cocomparability graphs are not a subclass of weakly triangulated graphs since, for every $k\geq 3$, the complement $\overline{C_{2k}}$ of a chordless cycle with length $2k$ is a cocomparability graph. 
Indeed, the complement of a $\overline{C_{2k}}$ (i.e.~the chordless cycle $C_{2k}$) can be transitively oriented. Therefore, our results do not follow from the paper of Dragan~\cite{Dragan97}.

More specifically, one of the main results of Dragan (see Theorem~1 in~\cite{Dragan97}) is that a graph $G$ is greedy matchable if and only if $G$ admits an \emph{admissible vertex ordering} (as defined in Definition~3 of~\cite{Dragan97}). 
Admissible orderings are characterized by the nine forbidden sub-orderings as shown in Figure~1 of Dragan~\cite{Dragan97}. 
However, three of these forbidden sub-orderings (namely the 2$^{nd}$, the 5$^{th}$, and the 9$^{th}$ one) are in fact LDFS umbrella-free orderings (see \cref{def:umbrella-free,def:ldfs-ordering} of our paper). 
To see this, observe that each of these three orderings 
(i)~is umbrella-free (see our \cref{def:umbrella-free}) and 
(ii)~does not contain any triple $a,b,c$ of vertices such that $a<_{\sigma}b<_{\sigma}c$, $\{a,c\}\in E$, 
and $\{a,b\}\notin E$ (see our \cref{def:good-bad-triple,def:ldfs-ordering}).

To illustrate this with an example, consider the graph $G=\overline{C_{6}}$ of our \cref{ldfs-fig-antihole} and recall from \Cref{subsec:vertex-orderings} that $\sigma = (a,e,c,f,b,d)$ is an LDFS umbrella-free vertex ordering of $G$. 
Note that this ordering $\sigma$ contains the orderings $(a,e,c,d)$ and $(a,e,c,b)$ as induced sub-orderings. 
Furthermore note that these sub-orderings correspond to the 2$^{nd}$ and the 5$^{th}$ forbidden sub-orderings of Figure~1 in Dragan's paper~\cite{Dragan97}, respectively. 
Hence, $\sigma$ is an example of an LDFS umbrella-free ordering which is not an admissible ordering; 
this is an alternative explanation of why our results do not follow from Dragan's paper~\cite{Dragan97}.
\end{remark}

In the remainder of this section, we show that the matching $M$ returned by RMM$(G,\sigma)$ is indeed a maximum matching of $G$. 
The proof is by contradiction and uses an appropriate \emph{potential function} $f$ that is defined over all matchings of $G$: 

\begin{definition}[potential function]
\label[definition]{def:potential}
	Let~${G = (V,E)}$ be a cocomparability graph and~$\sigma$ be an LDFS umbrella-free ordering of~$V = \{v_1, \ldots, v_n\}$ with~${v_1 <_\sigma \ldots <_\sigma v_n}$.
	Let~$M$ be~a matching of~$G$. 
	Then the potential function is~${f(M) := \sum_{i = 1}^n g_M(v_i)}$, where for each~$v_i\in V$:
	$$g_M(v_i) := \begin{cases}
						0, & \text{if }\{v_i,v_j\}\in M\text{ and }i<j, \\
						(i-j) \cdot (n+1)^i, & \text{if $\{v_i,v_j\}\in M$ and $j<i$}, \\
						i \cdot (n+1)^i, & \text{if }v_i\text{ is not matched within }M.
	            \end{cases}$$
\end{definition}

Note by~\cref{def:potential} that, for the empty matching, we have $f(\emptyset) = \sum_{i = 1}^n i \cdot (n+1)^i$. 
Then, as we add an edge $\{v_i,v_j\}$ to the current matching $M$, where $j<i$ and~$v_i$ and~$v_j$ are unmatched, we have that 
\begin{align*}
	f(M\cup \{\{v_i,v_j\}\}) 	& = f(M) - i(n+1)^i - j(n+1)^j + (i-j)(n+1)^i \\
							& = f(M) - j((n+1)^j + (n+1)^i) < f(M).
\end{align*}
Thus, adding edges to a matching decreases the potential function value.
The exponential dependency on the vertex-index in~$g_M$ ensures that matching vertices with higher index has a larger impact than matching vertices with lower index.
Furthermore, aiming at a small potential function value also means that the endpoints of the matched edges have only small index difference. 
We formalize this intuition in the next observation.
\begin{observation}
	\label{obs:potential-works-as-intended}
	Let~${G = (V,E)}$ be a graph and~$\sigma$ be an arbitrary ordering of~$V = \{v_1, \ldots, v_n\}$ with~${v_1 <_\sigma \ldots <_\sigma v_n}$.
	Let~$M$ and~$M'$ be two different matchings of~$G$ such that at~$v_i$ is the rightmost difference between~$M$ and~$M'$, that is, each vertex~$v_\ell$, $\ell > i$, is either free in both~$M$ and~$M'$ or matched with the same~$v_{\ell'}$ in both~$M$ and~$M'$.
	Suppose that: 
	\begin{itemize}
		\item $\{v_j,v_i\} \in M' \setminus M$, $j<i$, and
		\item $v_i$ is in~$M$ either free or matched to some~$v_{j'}$, $j' < j$.
	\end{itemize}
	Then, $f(M') < f(M)$.
\end{observation}
\begin{proof}
	We have
 	\begin{align*}
		f(M') - f(M) 	& = \sum_{k=1}^n g_{M'}(v_k) - g_{M}(v_k) = \sum_{i=k}^{i} g_{M'}(v_k) - g_{M}(v_k) 
	\end{align*}
	as by assumption~$v_i$ is the rightmost vertex where~$M$ and~$M'$ differ.
	Then, 
	\begin{align*}
		f(M') - f(M)	& = g_{M'}(v_i) - g_{M}(v_i) + \sum_{k=1}^{i-1} g_{M'}(v_k) - g_{M}(v_k) \\
						& <  (i-j)(n+1)^i - (i-x)(n+1)^i + \sum_{k=1}^{i-1} g_{M'}(v_k),
	\end{align*}
	where~$x=0$ if~$v_i$ is free in~$M$ or $x = j'$ if~$v_i$ is matched to~$v_{j'}$ in~$M$.
	In both cases we have~$j > x$ and thus
	\begin{align*}
		f(M') - f(M)	& < -(n+1)^i + \sum_{k=1}^{i-1} g_{M'}(v_k) < -(n+1)^i + \sum_{k=1}^{i-1} n(n+1)^k \\
						& = -(n+1)^i + n\frac{(n+1)^i-1}{n} - 1 \\
						& = -(n+1)^i + (n+1)^i-2 < 0. \\
	\end{align*}
\end{proof}

With the above observation the connection between the RMM algorithm and the potential function~$f$ is easy to see:

\begin{observation}
The matching $M$ returned by RMM$(G,\sigma)$ minimizes the function $f(M)$.
\end{observation}
\begin{proof}
	Let $M'$ be a matching such that $f(M')$ is minimum. 
	Consider the rightmost vertex $v_n$ in the ordering $\sigma$ and let $v_i$ be the rightmost neighbor of $v_n$ in $\sigma$. 
	Assume that $\{v_i,v_n\}\notin M'$. 
	Then let~$M''$ be the matching obtained by removing from~$M'$ any edges with endpoints~$v_i$ or~$v_n$, and by adding to it the edge~$\{v_i,v_n\}$. 
	By~\cref{obs:potential-works-as-intended}, we have $f(M'')<f(M')$, a contradiction.
	Thus $\{v_i,v_n\}\in M'$. 
	We can now recursively apply the same argument in the induced subgraph $G[(V \setminus \{v_i\}\setminus ) \{v_n\}]$, 
	which eventually implies that $M'$ is the matching returned by RMM$(G,\sigma)$.
\end{proof}

Before we prove our main result in~\cref{thm:rmm-returns-max-matching}, we need to prove a crucial technical lemma (\cref{lem:excluding-special-structure}). 
On a high level, our proof strategy is as follows: 
We consider a maximum matching~$M$ minimizing~$f$ and a matching~$M'$ produced by \cref{alg:RMM}.
If~$M=M'$, then we are done. 
Otherwise, we take the ``rightmost'' difference between~$M$ and~$M'$, that is the rightmost vertex~$v$ in~$M$ that is not matched in the same way in~$M'$ (or~$v$ is free in exactly one of the two matchings).
Then we show that matching~$v$ in~$M$ as in~$M'$ leads to another maximum matching~$M''$ such that~$f(M'') < f(M)$. 
To show this, we make a case distinction where in two cases we need to exclude the special scenario described in \cref{lem:excluding-special-structure}.
The existence of~$M''$ would be a contradiction to our choice of~$M$.
This shows that~$M=M'$.

In the next definition we introduce for every vertex $v$ the induced subgraph $G_{\sigma}(v)$ with respect to the ordering $\sigma$, which is fundamental for the statement and the proof of~\cref{lem:excluding-special-structure}.

\begin{definition}
	\label[definition]{def:left-of-x}
	Let $G=(V,E)$ be a cocomparability graph and let $\sigma=(v_1,v_2,\ldots,v_n)$ be an LDFS umbrella-free vertex ordering of~$G$. 
	Then, for every $v_i\in V$, the graph $G_{\sigma}(v_i)$ is the induced subgraph of~$G$ on the vertices $\{v_1,v_2,\ldots,v_i\}$.
\end{definition}

\begin{lemma}\label[lemma]{lem:excluding-special-structure}
	Let~$G = (V,E)$ be a cocomparability graph and~$\sigma$ be an LDFS umbrella-free ordering of~$V$.
	Let~$M$ be a maximum matching of $G$ such that $f(M)$ is minimum among all maximum matchings. 
	Then, there is no quadruple~$(a,b,c,x)$ of vertices in~$G$ satisfying all of the following six conditions:
	\begin{enumerate}
		\item $a <_\sigma b <_\sigma c \le_\sigma x$, \label[cond]{cond:ordering}
		\item $\{a,c\},\{b,c\}\in E$ and $\{a,b\} \notin E$, \label[cond]{cond:edges}
		\item $\{a,c\} \in M$,  \label[cond]{cond:matching-edge}
		\item there is no odd-length alternating path from~$a$ to~$b$ within~$G_{\sigma}(x)$, \label[cond]{cond:alternating-a-b-path}
		\item there is no odd-length alternating path from~$a$ to any free vertex~$v$ within~$G_{\sigma}(x)$, and \label[cond]{cond:alternating-a-d-path}
		\item there is no odd-length alternating path from~$b$ to any free vertex~$v$ within~$G_{\sigma}(x)$. \label[cond]{cond:alternating-b-d-path}
	\end{enumerate}
\end{lemma}

\begin{proof}
	Let~$G$, $\sigma$, and $M$ be as described in the statement of the lemma (see~\cref{fig:special-structure-for-lemma}).
	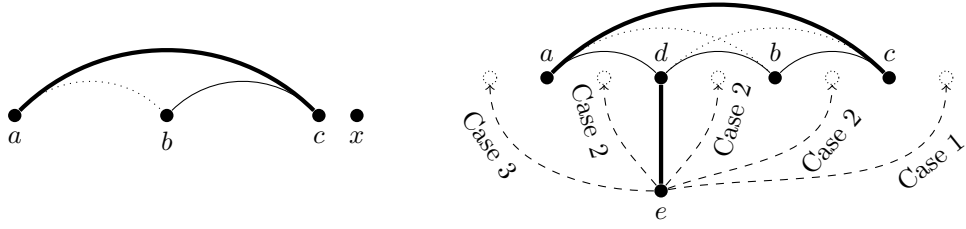
\begin{figure}
		\centering
		\begin{tikzpicture}

			\tikzstyle{node}=[minimum size=5pt,inner sep=0pt,fill=black,circle]
			\tikzstyle{edge}=[-,bend left=45]
			\tikzstyle{matched}=[-,ultra thick,bend left=45]
			\tikzstyle{nonedge}=[dotted,bend left=45]
			
			\node[node,label=below:$a$] at (0,0) (a) {};
			\node[node,label=below:$b$] at (2,0) (b) {};
			\node[node,label=below:$c$] at (4,0) (c) {};
			\node[node,label=below:$x$] at (4.5,0) (x) {};
			
			\draw[nonedge] (a) edge (b);
			\draw[matched] (a) edge (c);
			\draw[edge] (b) edge (c);

			\begin{scope}[xshift=7cm,yshift=0.5cm]
				\node[node,label=above:$a$] at (0,0) (a) {};
				\node[node,label=above:$d$] at (1.5,0) (d) {};
				\node[node,label=above:$b$] at (3,0) (b) {};
				\node[node,label=above:$c$] at (4.5,0) (c) {};

				\node[node,label=below:$e$] at (1.5,-1.5) (e) {};
				
				\draw[nonedge] (a) edge (b);
				\draw[matched] (a) edge (c);
				\draw[edge] (a) edge (d);
				\draw[edge] (b) edge (c);
				\draw[edge] (d) edge (b);
				\draw[nonedge] (d) edge (c);
				\draw[matched,bend left=0] (d) edge (e);

				\foreach \i / \pos / \out in { 1/5.25/10, 2/3.75/20, 2/2.25/50, 2/0.75/130, 3/-0.75/180 } {
					\node[minimum size=5pt,inner sep=0pt,draw=black,circle,densely dotted] at (\pos,0) (c\i) {};
					\draw[->,dashed] (e) to[out=\out, in=270, edge node={node [sloped,below,near end] {Case \i}}] (c\i) ;
				}
			\end{scope}
		\end{tikzpicture}
		\caption{Bold lines indicate matched edges, dotted lines indicate a non-edge. Left: Situation for invoking \cref{lem:excluding-special-structure} displaying condition \cref{cond:ordering,cond:edges,cond:matching-edge} where~$c \neq x$. Right: The case distinction in the proof of \cref{lem:excluding-special-structure} over the position of~$e$ in the order. }
		\label{fig:special-structure-for-lemma}
	\end{figure}
	The proof is by contradiction.
	Towards a contradiction let~$(a,b,c,x)$ be a quadruple of vertices satisfying all six conditions of the lemma. 
	Fix now vertex $x$. 
	Among all such quadruples with fixed $x$, let~$(a,b,c,x)$ be such that $a$ is leftmost in~$\sigma$, that is, for any other such quadruple~$(a',b',c',x)$ we have~$a \le_\sigma a'$. 
	Since~$\sigma$ is an LDFS ordering, it follows from \cref{cond:ordering,cond:edges} and \cref{def:good-bad-triple,def:ldfs-ordering} that there is a vertex~$d$ such that~$a <_{\sigma} d <_{\sigma} b$, $\{d,b\} \in E$, and $\{d,c\} \notin E$.
	Since~$\{d,c\} \notin E$ and $\sigma$ is umbrella-free, it follows that~$\{a,d\}\in E$.
	Observe that~$d$ is matched in~$M$ as otherwise \cref{cond:alternating-a-d-path,cond:alternating-b-d-path} would be violated.
	Thus, there is a vertex~$e \in V$ with~$\{e,d\} \in M$.
	Now we distinguish three cases with respect to the position of $e$ in the ordering $\sigma$.

	\textbf{Case 1: $c <_\sigma e$.} In this case we have that $a <_\sigma c <_\sigma e$, $\{d,e\}\in E$, and~$\{d,c\} \notin E$. Thus, since~$\sigma$ is umbrella-free, it follows that~$\{c,e\} \in E$.
	However, in this case for the matching~$M' = (M \setminus \{\{a,c\},\{d,e\}\}) \cup \{\{e,c\},\{a,d\}\}$ we invoke \cref{obs:potential-works-as-intended} with~$v_i = e$ to obtain~$f(M')<f(M)$, a contradiction.
	
	\textbf{Case 2: $a <_\sigma e <_\sigma c$.} If $\{a,e\}\in E$, then there exists the length-three alternating path 
	$(a,e,d,b)$ from $a$ to~$b$ within $G_{\sigma}(x)$, which is a contradiction to~\cref{cond:alternating-a-b-path}. Thus, $\{a,e\}\notin E$. Furthermore, $\{c,e\} \in E$, since $\sigma$ is umbrella-free and $\{a,c\} \in E$. Hence, for the matching~$M' = (M \setminus \{\{a,c\},\{d,e\}\}) \cup \{\{e,c\},\{a,d\}\}$ we invoke \cref{obs:potential-works-as-intended} with~$v_i = c$ to obtain~$f(M')<f(M)$, a contradiction.

	\textbf{Case 3: $e <_\sigma a$.} In this case it follows similarly to Case 2 that $\{a,e\}\notin E$ (proof by contradiction due to~\cref{cond:alternating-a-b-path}). 
	Furthermore observe that $\{e,d\},\{a,d\} \in E$ and~${\{e,d\} \in M}$. 
	Thus the triple~$(e,a,d)$ satisfies \cref{cond:ordering,cond:edges,cond:matching-edge}.
	Furthermore, if there exists an odd-length alternating path from~$e$ to~$a$ within $G_{\sigma}(x)$, then this alternating path can be extended through~$d$ to an odd-length alternating path from $a$ to~$b$ within $G_{\sigma}(x)$, which is a contradiction to \cref{cond:alternating-a-b-path}. 
	Hence there is no odd-length alternating path from~$e$ to~$a$ within $G_{\sigma}(x)$.
	Similarly, odd-length alternating paths from~$e$ (resp.~from~$a$) to a free vertex~$v$ within $G_{\sigma}(x)$ are excluded as well due to~\cref{cond:alternating-a-d-path} (resp.~due to~\cref{cond:alternating-b-d-path}). 
	Thus the quadruple~$(e,a,d,x)$ satisfies the six conditions of the lemma and it holds that~$e <_\sigma a$, a contradiction to the choice of the initial quadruple~$(a,b,c,x)$.
\end{proof}

We are now ready to prove our central result.

\begin{theorem}\label[theorem]{thm:rmm-returns-max-matching}
	For any $n$-vertex and $m$-edge cocomparability graph~$G$, \cref{alg:RMM} returns a maximum matching~$M$ of~$G$ in $O(n+m)$~time.
\end{theorem}

\begin{proof}
	Let~$G = (V,E)$ be a cocomparability graph, and let $\sigma$ be an umbrella-free LDFS ordering of $G$. 
	First we prove that \cref{alg:RMM} runs in $O(n+m)$ time. 
	To this end, we denote with~$\deg(v)$ the degree of a vertex~$v \in V$.
	During the execution of the algorithm we maintain the \emph{unvisited vertices} in a doubly linked list~$A$ (initially of size~$n$), according to their position in~$\sigma$. 
	Furthermore, we maintain for each vertex~$u$ its \emph{unvisited neighbors} in a doubly linked list~$N_u$ (initially of size~$\deg{(u)}$), again according to their position in~$\sigma$. 
	Once we have computed the ordering $\sigma$, the construction of the list~$A$ can be done in~$O(n)$ time. 
	The construction of all lists~$N_u$, where~$u\in V$, can be done in $O(n+m)$ time as follows. 
	We initialize~$N_u=\emptyset$ for every~$u\in V$. 
	Then we iterate for each vertex~$u\in V$ in the list~$A$ from left to right. 
	For every such vertex~$u$ we scan (in an arbitrary order) through its neighborhood~$N(u)$ (note that $N_u$ is at this point still incomplete), and for each $v\in N(u)$ we append vertex~$u$ in the list~$N_v$.

	\cref{line:init} can be clearly executed in $O(n)$ time. 
	The rightmost unvisited vertex $x$ in \cref{line:pick-rightmost-x} can be found in $O(1)$~time as the rightmost vertex in the list~$A$. 
	Once~$x$ is detected in \cref{line:pick-rightmost-x}, $x$ is removed from~$A$ also in $O(1)$~time. 
	Furthermore, $x$ is removed from all lists~$N_u$, where $\{x,u\}\in E$, in $O(\deg{(x)})$~time since~$x$ is always the last element in the respective list.
	Moreover, \cref{line:add-x} can be clearly executed in $O(1)$~time. 
	The if-condition of \cref{line:if-unvisited-neighbor} can be checked in $O(1)$ time by just checking whether the list~$N_x$ is empty. 
	Similarly to \cref{line:pick-rightmost-x}, \cref{line:pick-rightmost-y} can be executed in $O(\deg{(y)})$ time. 
	Furthermore, each of \cref{line:add-y,line:match-xy,line:return} can be clearly executed in $O(1)$ time. 
	Summarizing, the total running time of \cref{alg:RMM} is $O(n+\sum_{u\in V}\deg{(u)})=O(n+m)$. 
	
	For the correctness part, the proof is done by contradiction. 
	Let~$M$ be the matching returned by RMM$(G,\sigma)$.
	Assume towards a contradiction that~$M$ is not a maximum matching.
	For the rest of the proof, let~$M'$ denote a maximum matching that minimizes~$f(M')$ among all maximum matchings of $G$.
	Let~$x$ be the rightmost vertex in~$\sigma$ on which~$M$ differs from~$M'$. 
	Then $x$ is matched in at least one of the two matchings~$M$ and~$M'$. 
	Now we distinguish three cases with respect to the vertex that is matched with~$x$ in $M$ and $M'$.

	\textbf{Case 1: $x$ is matched in~$M'$ to some~$y \in V$ but is free in~$M$.} 
	Then $M$ and $M'$ also differ at vertex $y$. 
	Thus $y <_{\sigma} x$, since~$x$ is the rightmost vertex in which $M$ and~$M'$ differ. 
	Consider the iteration $t$ of \cref{alg:RMM} during which the algorithm visits $x$. 
	If~$y$ is free in~$M$, then this leads to a contradiction; 
	indeed, otherwise \cref{alg:RMM} would have matched~$x$ in iteration~$t$ as~$x$ has at least one unvisited neighbor, namely~$y$. %
	Hence, the vertex~$y$ is matched in~$M$ with a vertex~$z$ at an earlier iteration $t'<t$. 
	Then $M$ differs from~$M'$ also at vertex~$z$. 
	If $z<_{\sigma} x$, then \cref{alg:RMM} visits $x$ at an earlier iteration than~$z$, which is a contradiction to the assumption on~$z$. 
	Hence $x<_{\sigma} z$. 
	This is a contradiction to the assumption that $x$ is the rightmost vertex in $\sigma$ in which $M$ differs from~$M'$.

	\textbf{Case 2: $x$ is matched in~$M$ to some vertex~$y \in V$ but is free in~$M'$.} 
	If $y$ is free in~$M'$, then the matching $M'\cup \{\{x,y\}\}$ is larger than $M'$, which is a contradiction to the maximality assumption on $M'$. 
	Therefore $y$ is matched in~$M'$ to some vertex~$z \in V$. Note that $M$ and~$M'$ differ also on $y$ and $z$. Thus, it follows by the choice of~$x$ that~$y <_\sigma x$ and~$z <_\sigma x$. 
	Consider now the matching ${M'' := (M' \setminus \{\{y,z\}\}) \cup \{\{x,y\}\}}$, which is maximum since~${|M''|=|M'|}$. 
	However, invoking \cref{obs:potential-works-as-intended} with~$v_i = x$ yields~${f(M'') < f(M')}$, which is a contradiction to the assumption on the minimality of~$f(M')$.

	\textbf{Case 3: $\{x,y\} \in M$ and~$\{x,z\} \in M'$ with~$z \neq y$.}
	Then $M$ and~$M'$ differ also on $y$ and~$z$. Thus, it follows by the choice of~$x$ that~$y <_\sigma x$ and~$z <_\sigma x$. 
	Consider the iteration~$t$ of \cref{alg:RMM} during which the algorithm visits~$x$. 
	Suppose that vertex~$z$ is matched in~$M$ with a vertex~$p$ at an earlier iteration~$t'<t$. 
	Then $M$ differs from~$M'$ also at vertex~$p$. 
	If $p<_{\sigma} x$, then \cref{alg:RMM} visits~$x$ at an earlier iteration than~$p$, which is a contradiction to the assumption on~$p$. 
	If $x<_{\sigma} p$, then we have again a contradiction to the assumption that $x$ is the rightmost vertex in~$\sigma$ in which $M$ differs from~$M'$. 
	Thus $z$ is unmatched in $M$ at the iteration $t$ of \cref{alg:RMM} during which the algorithm visits~$x$. Furthermore, $z$ is also unvisited at iteration~$t$ since $z <_\sigma x$. 
	Now, if $y <_\sigma z$, then \cref{alg:RMM} would not match $x$ to $y$ at the execution of~\cref{line:pick-rightmost-y}, which is a contradiction. 
	Hence~$z <_\sigma y$.

	Suppose that~$y$ is free in~$M'$. 
	Then~${M'' := (M' \setminus \{\{x,z\}\}) \cup \{\{x,y\}\}}$ is another maximum matching.
	Invoking \cref{obs:potential-works-as-intended} with~$v_i = x$ yields~${f(M'') < f(M')}$, a contradiction to the assumption on~$M'$. 
	Hence, $y$ is matched in~$M'$ to some vertex~$w \in V$ with $w <_\sigma x$ by the choice of~$x$. 
	If~$\{w,z\}\in E$, then the matching~$M'' := (M' \setminus \{\{x,z\},\{w,y\}\}) \cup \{\{x,y\},\{z,w\}\}$ is another maximum matching. 
	Invoking \cref{obs:potential-works-as-intended} with~$v_i = x$ yields~${f(M'') < f(M')}$, which a contradiction to the choice of~$M'$. 
	Hence $\{z,w\} \notin E$.

	Suppose that within $G_{\sigma}(x)$ (\cref{def:left-of-x}) there exists an odd-length alternating path~$P_0$ with respect to~$M'$ from~$w$ to~$z$. 
	Let~$E_0$ be the edges in the path~$P_0$.
	Then, swapping in~$M'$ all edges on the path~$P_0$ (that is, replacing in~$M'$ the edges~$M'\cap E_0$ with the edges~$E_0 \setminus M'$), removing~$\{x,z\}$ and~$\{w,y\}$ from~$M'$, and adding~$\{x,y\}$ yields another maximum matching~$M''$. 
	Recall that $x$ is the rightmost vertex in which~$M$ and~$M'$ differ. 
	Thus, since the alternating path $P_0$ belongs to the induced subgraph $G_{\sigma}(x)$, it follows from \cref{obs:potential-works-as-intended} with~$v_i = x$ that~${f(M'') < f(M')}$, a contradiction to the choice of~$M'$. 
	Thus, within $G_{\sigma}(x)$ there exists no odd-length alternating path with respect to~$M'$ from~$w$ to~$z$.

	Similarly, suppose that within $G_{\sigma}(x)$ there exists an odd-length alternating path~$P_1$ with respect to~$M'$ from~$w$ (resp.~from~$z$) to a free vertex~$v$. 
	Then, swapping in~$M'$ all edges on the path $P_1$, removing~$\{x,z\}$ and~$\{w,y\}$ from~$M'$, and adding~$\{x,y\}$ yields another maximum matching~$M''$ for which \cref{obs:potential-works-as-intended} with~$v_i = x$ implies ~${f(M'') < f(M')}$, which is again a contradiction to the choice of~$M'$. 
	Thus there exists within $G_{\sigma}(x)$ no odd-length alternating path with respect to~$M'$ from~$w$ (resp.~from~$z$) to a free vertex~$v$.

	Now suppose that $w <_\sigma z$. 
	That is,~$w <_\sigma z <_\sigma y$, where $\{w,y\} \in E$ and~$\{z,w\}\notin E$. 
	Hence, $\{z,y\} \in E$ since~$\sigma$ is umbrella-free. 
	Thus, since $\{w,y\}\in M'$, it follows that the quadruple~$(w,z,y,x)$ satisfies all six conditions in the statement of \cref{lem:excluding-special-structure}. 
	This is a contradiction to \cref{lem:excluding-special-structure}, since $M'$ is assumed to be 
	a maximum matching of $G$ such that~$f(M')$ is minimum among all maximum matchings.

	Finally suppose that $z <_\sigma w$. 
	Recall that $M$ differs from $M'$ in $w$, since $\{y,w\}\in M'$ and $\{x,y\}\in M$. 
	Thus, since $x$ is the rightmost vertex in $\sigma$ in which $M$ differs from~$M'$, 
	it follows that $w <_\sigma x$. 
	That is,~$z <_\sigma w <_\sigma x$, where $\{x,z\} \in E$ and~$\{z,w\}\notin E$. 
	Hence, $\{w,x\} \in E$ since~$\sigma$ is umbrella-free. 
	Thus, since $\{x,z\}\in M'$, it follows that the quadruple~$(z,w,x,x)$ satisfies all six conditions in the statement of \cref{lem:excluding-special-structure}. 
	This is again a contradiction to \cref{lem:excluding-special-structure}, since $M'$ is assumed to be 
	a maximum matching of $G$ such that~$f(M')$ is minimum among all maximum matchings.

	Summarizing, the matching $M$ returned by \cref{alg:RMM} is a maximum matching.
\end{proof}

\section{Conclusion}
We presented a thorough mathematical analysis of an efficient and easy-to-implement linear-time greedy algorithm for computing maximum matchings on cocomparability graphs. 
This contributes to a long list of polynomial-time algorithms for problems on cocomparability graphs. 
Notably, most of this previous work showed polynomial-time (typically far from linear) algorithms for problems that are NP-hard on general graphs, while we improved a problem solvable in polynomial time on general graphs to linear time on cocomparability graphs.

Apart from being of interest on its own, our result might also be useful in a more general approach towards deriving faster algorithms for computing maximum matchings in relevant special cases. 
The fundamental idea behind this, as described in companion work~\cite{MNN17-MFCS}, is as follows.
First observe that, once a matching is given which has $k$~edges less than an optimal one, then using $k$~iterated augmenting path computations (each taking linear time~\cite{GT85}) one can improve it to a maximum matching. 
If for a graph~$G$ we also have a vertex subset set~$X$, $|X| = k$, such that~$G-X$ is a cocomparability graph, then for constant~$k$ we could get a linear-time algorithm for \Match as follows: 
First, delete the $k$~vertices from~$G$, then apply our linear-time algorithm, and then apply (as described above) at most $k$~iterations of augmenting path computations again with respect to the original graph~$G$, starting with the maximum matching for the cocomparability graph. 
A drawback of this approach is that we do not even know how to compute in linear-time a constant-factor approximation (which would be good enough) for the mentioned vertex deletion set of size~$k$. 
Hence, we consider it as an interesting challenge for future work to give a linear-time (constant-factor approximation) algorithm for computing a ``minimum-vertex-deletion-to-cocomparability'' set.
Based on the above considerations, for now we only can state the following result:

\begin{corollary}\label{thm:dist-to-cocomp-fptp}
        \Match can be solved in~$O(\distPara\cdot (n+m))$ time when given a size-$\distPara$ vertex set subset~$X$ such that deleting~$X$ from the given graph yields a cocomparability graph.
\end{corollary}

From a more general point of view, \cref{thm:dist-to-cocomp-fptp} is a contribution to the ``FPT in P'' program \cite{GMN17}, heading for more efficient polynomial-time algorithms based on problem parameterizations (also cf.~\cite{FLPSW18,FKMNNT2017,CDP18,BFNN17}).

\bigskip

\noindent\textbf{Acknowledgement.} We are deeply grateful to two anonymous reviewers for their very careful reading and their constructive feedback that helped us significantly clarify and improve the presentation.

\bibliographystyle{abbrvnat}
\bibliography{bib}

\end{document}